\documentclass[%
reprint,
nofootinbib,
 amsmath,amssymb,
 aps,
showkeys,
]{revtex4-2}

\usepackage{mathrsfs}

\usepackage{graphicx}
\usepackage{dcolumn}
\usepackage{bm}

\usepackage{amsthm}
\usepackage{xspace}
\usepackage[mathcal]{euscript}

\usepackage{orcidlink}

\usepackage{url}
\usepackage{hyperref}
\usepackage{color}
\definecolor{refcolor}{RGB}{0,0,190}
\hypersetup{
    colorlinks,
    citecolor=refcolor,
    filecolor=refcolor,
    linkcolor=refcolor,
    urlcolor=refcolor
}

\usepackage{booktabs}

\usepackage{bookmark}
\bookmarksetup{
  numbered, 
  open,
}

\begin{document}

\newtheorem{theorem}{Theorem}
\newtheorem{question}{Question}

\theoremstyle{remark}
\newtheorem{answer}{Answer}

\newtheorem{proofStep}{Step}
\numberwithin{proofStep}{theorem} 

\theoremstyle{definition}

\newtheorem{observation}{Observation}
\newtheorem{option}{Option}
\newtheorem{problem}{Problem}
\newtheorem{consequence}{Consequence}
\newtheorem{experiment}{Experiment}
\newtheorem{implication}{Implication}
\newtheorem{thoughtExp}{Thought Experiment}
\newtheorem{claim}{Claim}

\newtheorem{proposition}{Proposition}
\newtheorem{principle}{Principle}
\newtheorem{lemma}{Lemma}
\newtheorem{corollary}{Corollary}
\newtheorem{definition}{Definition}
\newtheorem{hypothesis}{Hypothesis}
\newtheorem{assumption}{Assumption}
\newtheorem{property}{Property}
\newtheorem{criterion}{Criterion}
\newtheorem{objection}{Objection}
\newtheorem{remark}{Remark}
\newtheorem{example}{Example}
\theoremstyle{remark}
\newtheorem{reply}{Reply}

\newenvironment{solution}[2] {\paragraph{Solution to {#1} {#2} :}}{\hfill$\square$}


\newcommand{\orcid}[1]{\href{https://orcid.org/#1}{\textcolor[HTML]{A6CE39}{\aiOrcid}}}

\def\bibsection{\section*{\refname}} 

\newcommand{\pbref}[1]{\ref{#1} (\nameref*{#1})}
   
\def\({\big(}
\def\){\big)}

\newcommand{\tn}{\textnormal}
\newcommand{\ds}{\displaystyle}
\newcommand{\dsfrac}[2]{\displaystyle{\frac{#1}{#2}}}

\newcommand{\boplus}{\textstyle{\bigoplus}}
\newcommand{\botimes}{\textstyle{\bigotimes}}
\newcommand{\bcup}{\textstyle{\bigcup}}
\newcommand{\bcap}{\textstyle{\bigcap}}
\newcommand{\bsqcup}{\textstyle{\bigsqcup}}
\newcommand{\bsqcap}{\textstyle{\bigsqcap}}

\newcommand{\MQS}{\ref{def:MQS}}
\newcommand{\TPS}{\ref{as:TPS}}
\newcommand{\ASSpace}{\ref{as:space}}
\newcommand{\ASBasis}{\ref{as:basis}}

\newcommand{\struct}{\mc{S}}
\newcommand{\kind}{\mc{K}}

\newcommand{\statespace}{\mathcal{S}}
\newcommand{\hilbert}{\mathcal{H}}
\newcommand{\vectorspace}{\mathcal{V}}
\newcommand{\mc}[1]{\mathcal{#1}}
\newcommand{\ms}[1]{\mathscr{#1}}

\newcommand{\wh}[1]{\widehat{#1}}
\newcommand{\dwh}[1]{\wh{\rule{0ex}{1.3ex}\smash{\wh{\hfill{#1}\,}}}}

\newcommand{\wt}[1]{\widetilde{#1}}
\newcommand{\wht}[1]{\widehat{\widetilde{#1}}}
\newcommand{\on}[1]{\operatorname{#1}}

\newcommand{\vect}[1]{\mathsf{#1}}
\newcommand{\oper}[1]{\wh{\mathbf{#1}}}
\newcommand{\TT}{\intercal}

\newcommand{\R}{\mathbb{R}}
\newcommand{\C}{\mathbb{C}}
\newcommand{\Z}{\mathbb{Z}}
\newcommand{\K}{\mathbb{K}}
\newcommand{\N}{\mathbb{N}}
\newcommand{\Prj}{\mathcal{P}}
\newcommand{\abs}[1]{\left|#1\right|}

\newcommand{\de}{\operatorname{d}}
\newcommand{\tr}{\operatorname{tr}}
\newcommand{\im}{\operatorname{Im}}

\newcommand{\dof}{d.o.f.\xspace}
\newcommand{\dofs}{d.o.f.s\xspace}

\newcommand{\ie}{\textit{i.e.}\ }
\newcommand{\vs}{\textit{vs.}\ }
\newcommand{\eg}{\textit{e.g.}\ }
\newcommand{\cf}{\textit{cf.}\ }
\newcommand{\etc}{\textit{etc}}
\newcommand{\etal}{\textit{et al.}}

\newcommand{\Span}{\tn{span}}
\newcommand{\pde}{PDE}
\newcommand{\U}{\tn{U}}
\newcommand{\SU}{\tn{SU}}
\newcommand{\GL}{\tn{GL}}

\newcommand{\schrod}{Schr\"odinger}
\newcommand{\vonneum}{Liouville-von Neumann}
\newcommand{\ks}{Kochen-Specker}
\newcommand{\leggarg}{Leggett-Garg}
\newcommand{\bra}[1]{\langle#1|}
\newcommand{\ket}[1]{|#1\rangle}
\newcommand{\braket}[2]{\langle#1|#2\rangle}
\newcommand{\ketbra}[2]{|#1\rangle\langle#2|}
\newcommand{\expectation}[1]{\langle#1\rangle}
\newcommand{\Herm}{\tn{Herm}}
\newcommand{\Eval}{\tn{Eval}}
\newcommand{\Sym}[1]{\tn{Sym}_{#1}}
\newcommand{\meanvalue}[2]{\langle{#1}\rangle_{#2}}

\newcommand{\btimes}{\boxtimes}
\newcommand{\btimess}{{\boxtimes_s}}

\newcommand{\h}{\mathbf{(2\pi\hbar)}}
\newcommand{\x}{\mathbf{x}}
\newcommand{\xThree}{\boldsymbol{x}}
\newcommand{\z}{\mathbf{z}}
\newcommand{\q}{\boldsymbol{q}}
\newcommand{\p}{\boldsymbol{p}}
\newcommand{\f}{\mathbf{f}}
\newcommand{\0}{\mathbf{0}}
\newcommand{\annih}{\widehat{\mathbf{a}}}

\newcommand{\cs}{\mathscr{C}}
\newcommand{\ps}{\mathscr{P}}
\newcommand{\xhat}{\widehat{\x}}
\newcommand{\phat}{\widehat{\mathbf{p}}}
\newcommand{\fqproj}[1]{\Pi_{#1}}
\newcommand{\cqproj}[1]{\wh{\Pi}_{#1}}
\newcommand{\cproj}[1]{\wh{\Pi}^{\perp}_{#1}}

\newcommand{\M}{\mathbb{E}_3}
\newcommand{\D}{\mathbf{D}}
\newcommand{\dn}{\tn{d}}
\newcommand{\db}{\mathbf{d}}
\newcommand{\n}{\mathbf{n}}
\newcommand{\m}{\mathbf{m}}
\newcommand{\V}[1]{\mathbb{V}_{#1}}
\newcommand{\F}[1]{\mathcal{F}_{#1}}
\newcommand{\Fvacuumfield}{\widetilde{\mathcal{F}}^0}
\newcommand{\nD}[1]{|{#1}|}
\newcommand{\Lin}{\mathcal{L}}
\newcommand{\End}{\tn{End}}
\newcommand{\vbundle}[4]{{#1}\to {#2} \stackrel{\pi_{#3}}{\to} {#4}}
\newcommand{\vbundlex}[1]{\vbundle{V_{#1}}{E_{#1}}{#1}{M_{#1}}}
\newcommand{\rep}{\rho_{\scriptscriptstyle\btimes}}

\newcommand{\intl}[1]{\int\limits_{#1}}

\newcommand{\moyalBracket}[1]{\{\mskip-5mu\{#1\}\mskip-5mu\}}

\newcommand{\Hint}{H_{\tn{int}}}

\newcommand{\quot}[1]{``#1''}

\def\sref #1{\S\ref{#1}}

\newcommand{\dBB}{de Broglie--Bohm}
\newcommand{\dBBt}{{\dBB} theory}
\newcommand{\pwt}{pilot-wave theory}
\newcommand{\PWT}{PWT}
\newcommand{\NRQM}{{\textbf{NRQM}}}

\newcommand{\image}[3]{
\begin{figure}[!ht]
\centering
\includegraphics[width=#2\textwidth]{#1}
\caption{\small{\label{#1}#3}}
\end{figure}
}

\newcommand{\no}[1]{\overline{#1}}

\newcommand{\cmark}{\ding{51}}%
\newcommand{\xmark}{\ding{55}}%

\hyphenation{pa-ram-e-trized ob-serv-er }



\title{Are observers reducible to structures?}

\author{Ovidiu Cristinel Stoica\ \orcidlink{0000-0002-2765-1562}}
\affiliation{
 Dept. of Theoretical Physics, NIPNE---HH, Bucharest, Romania. \\
	Email: \href{mailto:cristi.stoica@theory.nipne.ro}{cristi.stoica@theory.nipne.ro},  \href{mailto:holotronix@gmail.com}{holotronix@gmail.com}
	}%

\date{\today}

\begin{abstract}
Physical systems are characterized by their structure and dynamics. But the physical laws only express relations, and their symmetries allow any possible relational structure to be also possible in a different parametrization or basis of the state space. If observers were reducible to their structure, observer-like structures from different parametrizations would identify differently the observables with physical properties. They would perceive the same system as being in a different state. This leads to the question: is there a unique correspondence between observables and physical properties, or this correspondence is relative to the parametrization in which the observer-like structure making the observation exists?

I show that, if observer-like structures from all parametrizations were observers, their memory of the external world would have no correspondence with the facts, it would be no better than random guess. Since our experience shows that this is not the case, there must be more to the observers than their structure. This implies that the correspondence between observables and physical properties is unique, and it becomes manifest through the observers. This result is independent of the measurement problem, applying to both quantum and classical physics. It has implications for structural realism, philosophy of mind, the foundations of quantum and classical physics, and quantum-first approaches.
\end{abstract}


\maketitle

\section{Introduction}
\label{s:intro}

In this article I'll look into two apparently unrelated but strongly intertwined questions.
The first question is

\begin{question}
\label{question:correspondence-observables-meaning}
Is there an unambiguous correspondence between observables and physical properties?
\end{question}

Here by ``observable'' I mean the operator (on the quantum state space) or function (on the classical phase space) representing a physical property.

This question is relevant for various research programs aiming to recover the three-dimensional space and other physical properties and structures only from the quantum structure \cite{Carroll2021RealityAsAVectorInHilbertSpace,Giddings2019QuantumFirstGravity,Stoica2021SpaceThePreferredBasisCannotUniquelyEmergeFromTheQuantumStructure} or only from relations \cite{Rovelli1996RelationalQM}.

But we will see that this problem occurs even in standard Physics.
The physical laws alone don't give an unambiguous answer, because they only express relations. This makes them invariant with respect to a large group of reparametrizations of the state space.
In Classical Physics, these are the canonical transformations.
In Quantum Physics, they are unitary transformations.
This is similar to the invariance of the laws under transformations of spacetime as in Galilean and Special Relativity, where the conclusion is that even if space and time were absolute, as long as this doesn't transpire in the relation, all empirical predictions of relativity are correct.
So what gives physical meaning to observables?

The answer is given by the observers.
Observers make experiments, and establish a correspondence between observables and physical properties.
By ``observers'' I don't necessarily mean observers that ``collapse'' the wavefunction or play any role attributed to them to solve the measurement problem. In fact, the same problem appears in both Classical and Quantum Physics.

But the observers are physical systems, so they should also be subject to the physical laws.
This is often understood as implying that the observers should be completely reducible to their constituting relations, to their structure.
We will see that if this were true the observers would not be able to ascribe uniquely physical properties to the observables.
This leads us to revisiting a second question, whose answer is usually taken for granted,

\begin{question}
\label{question:observers-structure}
Are observers reducible to their structure?
\end{question}

The answer to Question~\ref{question:correspondence-observables-meaning} depends on the answer to Question~\ref{question:observers-structure}.
But how can we answer this question?

I will prove that, if the answer to Question~\ref{question:observers-structure} were affirmative, there would be no correlation between the observer's memory and the properties of external objects. In other words, observers would know nothing about the external world.
The reason boils down to the fact that the symmetry of the state space dissolves any such correlation, by allowing the external world to appear as having different properties in different parametrizations in which the structure of the observer is the same.
This doesn't happen if the only observers are the observer-like structures from a unique parametrization.

Section~\sref{s:quantum-physics} contains a review of the quantum formalism that will be applied in the rest of the article.
Section~\sref{s:symmetry-law} reviews the symmetry transformations of the state space and the physical laws.
Section~\sref{s:meta-relativity-observers} shows how these symmetries imply that the same structures can occur in any parametrization of the state space, and that the same state can appear as different structures in different parametrizations.
I prove the main result in Section~\sref{s:non-reducibility}.
Some physical aspects of the proof are discussed in Section~\sref{s:discussion}.
The article concludes with a discussion of some of the implications of the result in Section~\sref{s:implications}.

\section{Physical laws and structures}
\label{s:quantum-physics}

Any structure of any system in the world, including the structure of the observers, is presumably a physical structure.
Therefore, here I review the physical structures and their dynamics, and the formalism used in the article.

I will use the quantum formalism, because the world is quantum, and all existing systems are ultimately quantum. But Classical Physics can be expressed in the same formalism \cite{Koopman1931HamiltonianSystemsAndTransformationInHilbertSpace}, and we will obtain the same result.

The state of a quantum system is represented by a \emph{state vector}, a complex vector $\ket{\psi}$ of unit length in an infinite-dimensional vector space $\hilbert$ called here \emph{state space}. The state space is a \emph{Hilbert space}, a special kind of vector space endowed with a complex scalar product with the Hermitian property $\braket{\psi}{\psi'}=(\braket{\psi'}{\psi})^\ast$, where $\ast$ is the complex conjugation.
The length of a vector $\ket{\psi}$ from $\hilbert$ is $\abs{\ket{\psi}}:=\sqrt{\braket{\psi}{\psi}}$.
If $\oper{A}:\hilbert\to\hilbert$ is a linear operator, there is a unique linear operator $\oper{A}^\dagger$, called the \emph{adjoint} of $\oper{A}$, so that $\bra{\psi}\oper{A}\ket{\psi'}=\(\bra{\psi'}\oper{A}^\dagger\ket{\psi}\)^\ast$ for all $\ket{\psi},\ket{\psi'}\in\hilbert$.

But since all unit vectors are identical under the symmetries of the state space, the vector alone is not sufficient to describe the structure and properties of a system.
The properties are represented by linear operators $\oper{A}:\hilbert\to\hilbert$, that are \emph{Hermitian}, \ie $\oper{A}^\dagger=\oper{A}$.
We call them \emph{observables}.
The property represented by $\oper{A}$ has a definite value $a$ for a state vector $\ket{\psi}$ if and only if $\ket{\psi}$ is an \emph{eigenvector} of $\oper{A}$ with the \emph{eigenvalue} $a$, that is,
\begin{equation}
\label{eq:eigen}
\oper{A}\ket{\psi}=a\ket{\psi}.
\end{equation}
Since $\oper{A}$ is Hermitian, it has only real eigenvalues.

In Quantum Physics, for a given state $\ket{\psi}$, only some of the properties have definite values -- any property $A$ whose observable $\oper{A}$ satisfies equation \eqref{eq:eigen} for $\ket{\psi}$ and some eigenvalue $a$, which is the value of $A$.

The state space admits a special basis, consisting of vectors uniquely identified by the combination of definite values of positions, components of spin, and internal degrees of freedom (\dofs) of different particles.
(I use the word ``particles'' because I work in the ``particle representation'' based on positions, but these are wavefunctions, not point-particles.)
What physical property each of these \dofs represents is very important, but for simplicity I will denote these values uniformly by $q_1,q_2,\ldots$, and the basis vectors by $\ket{q_1,q_2,\ldots}$.
The values $q_1,q_2,\ldots$ are eigenvalues of the operators $\wh{q}_1,\wh{q}_2,\ldots$, representing positions, components of spin, and internal \dofs.
Each basis vector satisfies the equation
\begin{equation}
\label{eq:basis-eigen}
\wh{q}_j\ket{q_1,q_2,\ldots}=q_j\ket{q_1,q_2,\ldots}
\end{equation}
for each of these operators $\wh{q}_j$ and the corresponding eigenvalue $q_j$.
The number and kinds of the operators $\wh{q}_j$ required to ``fill the slots'' of a vector $\ket{q_1,q_2,\ldots}$ depend on the number and type of particles whose state is represented by $\ket{q_1,q_2,\ldots}$.

All possible combinations of eigenvalues $(q_1,q_2,\ldots)$ of the operators $(\wh{q}_1,\wh{q}_2,\ldots)$ form a \emph{parameter space} $\mc{C}$, usually called \emph{position configuration space}.
Since the number and kind of the operators $\wh{q}_j$ depend on the number and type of particles, $\mc{C}$ is not a connected manifold, but a union of manifolds of various dimensions, each of them being the parameter space for systems of different numbers and types of particles with definite values for the spin and internal \dofs.

The \emph{wavefunction} for a state vector $\ket{\psi}\in\hilbert$ is a complex function $\psi:\mc{C}\to\C$ defined by
\begin{equation}
\label{eq:wavefunction}
\psi(q_1,q_2,\ldots):=\braket{q_1,q_2,\ldots}{\psi}.
\end{equation}

With the notation $\q=(q_1,q_2,\ldots)$, the wavefunction is
\begin{equation}
\label{eq:wf}
\psi(\q)=\braket{\q}{\psi}.
\end{equation}

But how exactly are the other physical properties represented by operators?
They all depend of the operators $\wh{q}_j$ and the \emph{momentum operators} $\wh{p}_k:=-i\hbar\frac{\partial}{\partial q_k}$ canonically conjugate to those $\wh{q}_k$ that are position operators, where $\hbar$ is the reduced Planck's constant. The operators $\wh{p}_k$ commute with one another and with $\wh{q}_j$ for $j\neq k$.
All other physical properties depend on the operators $\wh{q}_j$ and $\wh{p}_k$, so they are represented by operators of the form $\wh{f}(\wh{\q},\wh{\p})$, where $\wh{\q}=(\wh{q}_1,\wh{q}_2,\ldots)$ and $\wh{\p}=(\wh{p}_1,\wh{p}_2,\ldots)$.

The \emph{evolution equation} is, for any $\ket{\psi(0)}$ and any $t\in\R$,
\begin{equation}
\label{eq:unitary_evolution}
\ket{\psi(t)}=\oper{U}_t\ket{\psi(0)}.
\end{equation}
The \emph{evolution operators} $\oper{U}_t$ are defined by $\oper{U}_t=e^{-\frac{i}{\hbar}\oper{H}t}$, where $\oper{H}$ is the \emph{Hamiltonian operator}.
The operators $\oper{U}_t$ are \emph{unitary operators}, that is, they preserve the structure of the state space $\hilbert$, including the scalar product. An operator $\oper{U}$ is unitary if and only if $\oper{U}^{-1}=\oper{U}^\dagger$.
Equation \eqref{eq:unitary_evolution} gives the solutions of the {\schrod} equation
\begin{equation}
\label{eq:schrod}
i\hbar\frac{d}{dt}\ket{\psi(t)}=\oper{H}\ket{\psi(t)}.
\end{equation}

The propagation of the wavefunction on the parameter space $\mc{C}$ is therefore given by $\psi(\q,t)=\braket{\q}{\psi(t)}$.

The {\schrod} equation expressed in terms of the parameters from $\mc{C}$, that is, in the basis of position eigenvectors $\(\ket{q_1,q_2,\ldots}\)_{(q_1,q_2,\ldots)\in\mc{C}}$, has the following form in terms of the Hamiltonian matrix $\bra{\q}\oper{H}\ket{\wt{\q}}$
\begin{equation}
\label{eq:schrod-position}
i\hbar\frac{d}{dt}\braket{\q}{\psi(t)}
=\int_{\wt{\q}\in\mc{C}}\bra{\q}\oper{H}\ket{\wt{\q}}\braket{\wt{\q}}{\psi(t)}d\wt{\q}.
\end{equation}

\begin{remark}[Physical structures]
\label{rem:metaobserver}
The wavefunction $\psi(\q,t)$ contains the complete information about the system. This is true even if we interpret $\psi(\q,t)$ probabilistically and include collapses in its evolution.
Therefore, an omniscient being somewhat similar to the Laplace demon, let's call it the \emph{metaobserver}, should be able to read everything about the system by examining the patterns of $\psi(\q,t)$. 
The structure of any system is manifest in these patterns. 
For example, the structure of a separable state of two particles is different from that of an entangled state, and this is visible in the pattern of the wavefunction of the separable state because the wavefunction is factorizable, \ie the variables $q_j$ of the first particle are separable from those of the second particle.
\qed
\end{remark}

\begin{remark}[Physical space]
\label{rem:space}
For a particle, we can identify as physical space the space parametrized by $(q_1,q_2,q_3)$ and finitely many additional parameters for the spin and internal \dofs.
If there is another particle with position space $(q_4,q_5,q_6)$, and the interactions between them depend on the distance $\sqrt{\abs{q_4-q_1}^2+\abs{q_5-q_2}^2+\abs{q_6-q_3}^2}$, this determines an identification between the spaces $(q_1,q_2,q_3)$ and $(q_4,q_5,q_6)$. This identification extends to how many particles exist, as long as they participate in such interactions. In nonrelativistic Quantum Mechanics, this works because the potential depends on the distance \cite{Albert1996ElementaryQuantumMetaphysics}.
In Quantum Field Theory fields interact by local coupling, which allows the identification between the spaces of the fields.
Also see \cite{Stoica2021SpaceThePreferredBasisCannotUniquelyEmergeFromTheQuantumStructure}.
\qed
\end{remark}

\section{Symmetries of the physical laws}
\label{s:symmetry-law}

When we examine the wavefunction on the parameter space $\mc{C}$, we assume a preferred basis for $\hilbert$, consisting of vectors of the form $\ket{q_1,q_2,\ldots}$.
But it is possible to change this basis by applying a \emph{unitary transformation} (a linear transformation that preserves the structure of the state space).
A unitary transformation $\oper{S}$ transforms any operator $\oper{A}$ into another operator $\oper{A}'=\oper{S}\oper{A}\oper{S}^\dagger$.

In particular, we get $\wh{q}'_j=\oper{S}\wh{q}_j\oper{S}^\dagger$ and $\wh{p}'_k=\oper{S}\wh{p}_k\oper{S}^\dagger$.
The eigenvalues of the operators $\wh{q}'_j=\oper{S}\wh{q}_j\oper{S}^\dagger$ form a parameter space $\mc{C}'$, in general different from $\mc{C}$, which I will denote by $\oper{S}(\mc{C})$.
If the transformation $\oper{S}$ represents a change of the reference frame in space or spacetime, $\oper{S}(\mc{C})=\mc{C}$, although the parametrization will be different. But in general $\oper{S}(\mc{C})\neq\mc{C}$.

Any operator $\wh{f}(\wh{\q},\wh{\p})$ is transformed into an operator
\begin{equation}
\label{eq:transform-q-prop}
\wh{f}'(\wh{\q},\wh{\p})=\oper{S}\wh{f}(\wh{\q},\wh{\p})\oper{S}^\dagger=\wh{f}(\wh{\q}',\wh{\p}'),
\end{equation}
where $\wh{f}(\wh{\q}',\wh{\p}')$ depends functionally on $\wh{\q}'$ and $\wh{\p}'$ in the same way $\wh{f}(\wh{\q},\wh{\p})$ depends on $\wh{\q}$ and $\wh{\p}$.

By changing the basis with $\oper{S}$, the {\schrod} equation \eqref{eq:schrod} is transformed into 
$i\hbar\frac{d}{dt}\ket{\psi'(t)}=\oper{H}'\ket{\psi'(t)}$, where $\ket{\psi'(t)}=\oper{S}^\dagger\ket{\psi(t)}$ and $\oper{H}'=\oper{S}^\dagger\oper{H}\oper{S}$.
Therefore, a quantum system remains a quantum system under symmetry transformations.

In the new basis $\(\ket{q'_1,q'_2,\ldots}\)_{(q'_1,q'_2,\ldots)\in\mc{C}'}$ resulting from changing the basis with the transformation $\oper{S}$, the components of the state vector $\ket{\psi}$ are $\psi'(\q',t)=\braket{\q'}{\psi(t)}=\bra{\q}\oper{S}^\dagger\ket{\psi}$, the Hamiltonian matrix is $\bra{\q}\oper{S}^\dagger\oper{H}\oper{S}\ket{\wt{\q}}$, and equation \eqref{eq:schrod-position} takes the form
\begin{equation}
\label{eq:schrod-new-basis}
i\hbar\frac{d}{dt}\braket{\q'}{\psi(t)}
=\int_{\wt{\q}'\in\mc{C}'}\bra{\q'}\oper{H}\ket{\wt{\q}'}\braket{\wt{\q'}}{\psi(t)}d\wt{\q}',
\end{equation}
which is equivalent with
\begin{equation}
\label{eq:schrod-new-basis-old}
i\hbar\frac{d}{dt}\bra{\q}\oper{S}^\dagger\ket{\psi(t)}
=\int_{\wt{\q}\in\mc{C}}\bra{\q}\oper{S}^\dagger\oper{H}\oper{S}\ket{\wt{\q}}\bra{\wt{\q}}\oper{S}^\dagger\ket{\psi(t)}d\wt{\q}.
\end{equation}

We see that the {\schrod} equation and the Hamiltonian matrix have the same form in the new basis if and only if $\oper{S}\oper{H}=\oper{H}\oper{S}$.
This is a special type of symmetry transformation:
\begin{definition}
\label{def:isonomy}
An \emph{isonomy} (or \emph{structural symmetry transformation}) is a symmetry transformation $\oper{S}$ that preserves the form of the physical laws.

If $\oper{S}$ is an isonomy, we say that $\oper{S}\oper{A}\oper{S}^\dagger$ is \emph{isonomic} with $\oper{A}$, and that $\oper{S}(\mc{C})$ is isonomic with $\mc{C}$.
\end{definition}

A transformation $\oper{S}$ is a structural symmetry transformation if and only if it is a symmetry of the Hamiltonian, \ie it commutes with the Hamiltonian operator $\oper{H}$, \ie $\oper{S}\oper{H}=\oper{H}\oper{S}$. Then, $\oper{S}\oper{H}\oper{S}^\dagger=\oper{H}$, so indeed the dynamics from \eqref{eq:schrod-new-basis-old} follows the same law as that from \eqref{eq:schrod-position}.

\begin{remark}
\label{rem:identical-structures}
For a state vector $\ket{\psi}$, the wavefunction $\braket{\q'}{\psi}$ on $\mc{C}'$ contains, in general, different structures than the wavefunction $\braket{\q}{\psi}$ on $\mc{C}$.
For example, plane waves in the position space appear as wavefunctions concentrated at points in the momentum space (the Fourier transform of the position space), and vice versa.
Only Gaussian wavefunctions appear as Gaussians in both the position space and the momentum space, though not in all other parameter spaces.
Even if $\oper{S}\oper{H}=\oper{H}\oper{S}$, the structures on $\mc{C}'$ are in general very different from those on $\mc{C}$.

Similarly, the wavefunction $\braket{\q'}{\psi'}$ on $\mc{C}'$, corresponding to another state vector $\ket{\psi'}=\oper{S}\ket{\psi}$,  looks exactly like the wavefunction $\psi$ on $\mc{C}$. This is because the unitary operator $\oper{S}$ preserves the scalar product, so the scalar product between $\ket{\q'}=\oper{S}\ket{\q}$ and $\ket{\psi'}=\oper{S}\ket{\psi}$ is equal, for all $\q\in\mc{C}$, to the scalar product between $\ket{\q}$ and $\ket{\psi}$, 
\begin{equation}
\label{eq:structures-change-basis}
\braket{\q'}{\psi'}=\bra{\q}\oper{S}^\dagger\oper{S}\ket{\psi}=\braket{\q}{\psi}.
\end{equation}

Then, any structure possible on the parameter space $\mc{C}$ is also possible on any other parameter space $\oper{S}(\mc{C})$.
\qed
\end{remark}

\section{Meta-Relativity of observers}
\label{s:meta-relativity-observers}

We attribute a special, unique physical meaning to the operators $\wh{q}_j$.
But due to its symmetries, the quantum formalism alone doesn't distinguish them from other choices $\wh{q}'_j$.
It doesn't allow any parameter space $\mc{C}$ to be special among all possible parameter spaces $\oper{S}(\mc{C})$.
What breaks the huge symmetry of the state space?

We attribute physical meaning to the operators because of the experiments.
But how would this work?

Suppose that the wavefunction $\braket{\q}{\psi}$ on $\mc{C}$ contains observer-like structures performing experiments. If they are conscious, they would consider that the parameters $q_j$ correspond to positions in the physical space.

Let $\mc{C}'=\oper{S}(\mc{C})$ be another parameter space, parametrized by the eigenvalues of $\wh{q}'_j=\oper{S}\wh{q}_j\oper{S}^\dagger$. On $\mc{C}'$ there are wavefunctions $\braket{\q'}{\psi'}$ that also contain observer-like structures.
For example, if $\ket{\psi'}=\oper{S}\ket{\psi}$, the structures from $\braket{\q'}{\psi'}$ are identical to those from $\braket{\q}{\psi}$, and if $\oper{S}$ commutes with $\oper{H}$, they evolve identically.

Then, if observers are reducible to the structures, no experiment can tell them in what parameter space they live. It could be $\mc{C}$ or any other parameter space $\oper{S}(\mc{C})$.
\begin{proposition}
\label{thm:alternative-positions}
Observer-like structures from different parameter spaces identify different physical spaces.
\end{proposition}
\begin{proof}
Observer-like structures from different parameter spaces $\mc{C}$ and $\mc{C}'$ use different sets of operators $\wh{\q}$ and $\wh{\q}'$ to represent positions in space.
Then, the identification of the physical space from Remark~\ref{rem:space} leads to different results.
The parameter spaces $\mc{C}$ and $\mc{C}'$ coincide only if the operators $\wh{\q}'$ commute with all $\wh{\q}$, and this happens when they are all functions of $\wh{\q}$ independent of $\wh{\p}$.
Otherwise, from Remark~\ref{rem:space}, the resulting physical spaces will be different.
\end{proof}

Consequently, they perceive the same state as differently structured with respect to the physical space.

\begin{remark}
\label{rem:concealed-space}
Note that even if there is an objectively unique physical space, even if its associated observables are ontologically special, whatever this means, Proposition~\ref{thm:alternative-positions} implies that observer-like structures from $\mc{C}$ and $\mc{C}'$ still identify different physical spaces.
What is physical space to an observer-like structure on $\mc{C}$, to an observer-like structure on $\mc{C}'$ it appears as a space consisting of other three \dofs, associated to different physical properties.
And the same happens for all other observables.
\qed
\end{remark}

Let's extract these findings in the form of a principle.

\begin{principle}[Meta-Relativity]
\label{pp:meta-relativity-structure}
Observer-like structures on any two parameter spaces $\mc{C}$ and $\mc{C}'$ agree upon the laws of physics if and only if $\mc{C}$ and $\mc{C}'$ are isonomic.
But in general they disagree about the physical properties associated with the observables and about the physical space.
\end{principle}

Neither the relations that we can extract from experiments nor the theory can determine the physical meaning of the operators. The physical meaning of the operators is relative to the parameter space, in the sense that observers from a parameter space $\mc{C}$ have a different physical interpretation of the operators compared to the observers from another parameter space $\mc{C}'$.
But all observers from the same parameter space agree upon the physical meaning of the operators.

Principle~\ref{pp:meta-relativity-structure} is very similar to the Principle of Relativity. For isometric coordinate transformations in space from $(x,y,z)$ to $(x',y',z')$, different observers agree upon the physical laws. But there is no way to tell that the coordinates $(x,y,z)$ are special compared to $(x',y',z')$.
The Poincar\'e transformations that appear in Special Relativity are particular structural symmetry transformations.

Principle~\ref{pp:meta-relativity-structure} extends the Principle of Relativity, so I chose to name it ``the Principle of \nameref{pp:meta-relativity-structure}''.
But it extends the Principle of Relativity only for structures, not for their physical meaning. Principle of Relativity remains true about the physical meaning of spacetime.

All we can access by \emph{intersubjectively verifiable experiments} are the relations. Relations allow us to build mathematical models of the world. 
The nature of the relata is outside the realm of intersubjectively verifiable experiments.
This truth was noticed in one form or another by various philosophers, notably Poincar\'e \cite{Poincare2022ScienceAndHypothesis} and Russell \cite{Russell1927AnalysisOfMatter}, and it is called \emph{epistemic structural realism} \cite{sep-structural-realism}.
Epistemic structural realism seems to apply to science, because 
\begin{enumerate}
	\item 
No intersubjectively verifiable experiment can go beyond the relations.
	\item 
No theoretical model can go beyond relations. Logically consistent theories admit faithful mathematical models in terms of mathematical structures \cite{ChangAndKeisler1990ModelTheory,Hodges1997ShorterModelTheory}. But mathematical structures themselves are nothing but sets and relations \cite{Gratzer2008UniversalAlgebra}, and the elements of the sets are characterized exclusively by the relations in which they participate.
\end{enumerate}

\begin{remark}
\label{rem:concealed-basis}
Even if $(\wh{q}_1,\wh{q}_2,\ldots)$ have a special ontic status compared to other choices $(\wh{q}'_1,\wh{q}'_2,\ldots)$, this can't be assessed by intersubjectively verifiable empirical means.
These means can only establish that different observer-like structures from the same parameter space assign the same physical properties to the observables.

If there were objective means to determine $(\wh{q}_1,\wh{q}_2,\ldots)$ as the preferred properties, then we would be able to obtain them from the abstract state vector $\ket{\psi}$ and the abstract Hamiltonian $\oper{H}$, where ``abstract'' means that they are not expressed in a preferred basis or parametrization.
In the case of the Hamiltonian, this means its spectrum, including the multiplicities.
But it's impossible to obtain them uniquely \cite{Stoica2021SpaceThePreferredBasisCannotUniquelyEmergeFromTheQuantumStructure,Stoica2024DoesTheHamiltonianDetermineTheTPSAndThe3dSpace}.

We can claim that the properties represented by the operators $(\wh{q}_1,\wh{q}_2,\ldots)$ associated to our parameter space are ontologically ``more real'' than other possible choices $(\wh{q}'_1,\wh{q}'_2,\ldots)$. But this doesn't solve the problem, since all operators $\wh{f}(\wh{\q},\wh{\p})$ are ``as real'' as $\wh{q}_j$ and $\wh{p}_k$.
In particular, any other choice $(\wh{q}'_1,\wh{q}'_2,\ldots)$ is ``as real'' as $(\wh{q}_1,\wh{q}_2,\ldots)$, because any $\wh{q}'_j$ is also a function $\wh{q}'_j=\wh{q}'_j(\wh{\q},\wh{\p})$.
\qed
\end{remark}

\begin{remark}
\label{rem:concealed-variables}
Even if there are some unknown hidden structures that break the symmetry so that $(\wh{q}_1,\wh{q}_2,\ldots)$ are special compared to $(\wh{q}'_1,\wh{q}'_2,\ldots)$ the problem remains. If we don't have empirical access to those hidden structures, they are useless. And if we will be able to access them, we will have to complete our theory. But the new theory will also be relational, so it will have its own large symmetry group, hence it will have the same problem.

For example, consider that we add point-particles with definite positions, as in the \emph{pilot-wave theory} \cite{BohmHiley1993UndividedUniverse}.
Then, the symmetry transformations are canonical transformations of the phase space of point-particles, done in tandem with unitary transformations of the pilot wave.
The classical system of point-particles can be described using the quantum formalism (see Corollary~\ref{thm:conscious-observer-classical}), coupled with the quantum system of the pilot wave.
Therefore, the quantum formalism and the discussions from this article still apply, and adding point-particles can't avoid the conclusion.
\qed
\end{remark}

\section{Are observers reducible to structures?}
\label{s:non-reducibility}

Given that there is no physical way to determine if an observer-like structure is special compared to an isomorphic structure from another parameter space, Principle of \nameref{pp:meta-relativity-structure} leads to the following question:
\begin{question}
\label{question:all-structures}
Are the observer-like structures from all parameter spaces conscious?
\end{question}

A more specific variant of Question~\ref{question:all-structures} is the following one. If the wavefunction on a parameter space $\mc{C}$ contains observer-like structures that are conscious, and if isomorphic structures can be found on another parameter space $\mc{C}'$, are the latter structures conscious as well?

If consciousness is reducible to structures, the Principle of \nameref{pp:meta-relativity-structure} implies that all observer-like structures in all parameter spaces are conscious.
But if it turns out that not all observer-like structures in all parameter spaces are conscious, this would mean that there is more to consciousness than just the structure.

\begin{theorem}
\label{thm:conscious-observer}
If observers were reducible to structures, they would know nothing about the external world.
\end{theorem}
\begin{proof}
I will prove the theorem in three steps:

Step~\ref{step:thm:conscious-observer:anything}. Show that for any observer-like structure whose memory is correlated with its environment, there are parameter spaces on which the wavefunction of the same state contains identical observer-like structures, but whose environment can appear to be in any possible state.
In particular, if an observer-like structure knows that the value of a property of the environment is $a$ and it really is $a$, there is an identical observer-like structure with the same structure, but for whom that property of the environment can have any other possible value $a'$.

Step~\ref{step:thm:conscious-observer:preserve}. Show that the transformation connecting the two parameter spaces can be chosen to commute with the total Hamiltonian, so that the evolution law for the two observer-like structures is the same.

Step~\ref{step:thm:conscious-observer:knows-nothing}. Show that the alternative parameter spaces with different environments are uniformly distributed, so that there is no correlation between the memory of a generic observer-like structure and its environment.

The observer is a subsystem of the universe. We can represent the state of the universe by
\begin{equation}
\label{eq:observer-state}
\ket{\psi}=\ket{\psi_{\omega}}\ket{\psi_{\varepsilon}},
\end{equation}
where $\ket{\psi_{\omega}}\in\hilbert_{\omega}$ represents the observer, $\ket{\psi_{\varepsilon}}\in\hilbert_{\varepsilon}$ represents the rest of the world, and $\hilbert\cong\hilbert_{\omega}\otimes\hilbert_{\varepsilon}$.
The following proof can be adapted easily for an observer entangled with the environment.

The parameter space decomposes as a Cartesian product $\mc{C}\cong\mc{C}_{\omega}\times\mc{C}_{\varepsilon}$, on which the wavefunction is
\begin{equation}
\label{eq:observer-wf}
\psi(\q_{\omega},\q_{\varepsilon})=\psi_{\omega}(\q_{\omega})\psi_{\varepsilon}(\q_{\varepsilon}),
\end{equation}
where $\q_{\omega}\in\mc{C}_{\omega}$ and $\q_{\varepsilon}\in\mc{C}_{\varepsilon}$.

Since the observer can only access directly her own present state of mind, it is sufficient that the state vector $\ket{\psi_{\omega}}$ represents the brain of the observer.
However, to humor anyone who would object to this, let us assume that $\ket{\psi_{\omega}}$ represents a more extended system that contains the observer. For example, suppose that $\ket{\psi_{\omega}}$ represents a room in which the observer presently is.

\begin{proofStep}
\label{step:thm:conscious-observer:anything}
The observer's memory contains information about various properties of various external systems.
For example, if she remembers that in the corner of her kitchen there is a table, her memory contains information about the approximate size of the kitchen and the table, and their relative position.
So our observer expects that, if she goes to the kitchen, she will find these to be true.
Since the state of the external objects is characterized by their physical properties, let us choose such a property, represented by an observable $\oper{A}_{\varepsilon}$ (seen as an operator on the entire space $\hilbert$), so that $\ket{\psi}$ is an eigenvector of $\oper{A}_{\varepsilon}$ with eigenvalue $a$.

Now consider another state $\ket{\psi'}=\ket{\psi_{\omega}}\ket{\psi_{\varepsilon}'}$, in which the observer has the exact same structure, but so that $\ket{\psi'}$ is an eigenvector of $\oper{A}_{\varepsilon}$ with a different eigenvalue $a'$.
The wavefunction of $\ket{\psi'}$ on the parameter space $\mc{C}$ is $\psi'(\q_{\omega},\q_{\varepsilon})=\psi_{\omega}(\q_{\omega})\psi'_{\varepsilon}(\q_{\varepsilon})$.
According to  Remark \ref{rem:identical-structures}, there is a unitary transformation $\oper{S}$ that maps $\ket{\psi}$ into $\ket{\psi'}$.
On the resulting parameter space $\mc{C}':=\oper{S}\(\mc{C}\)$, $\ket{\psi'}$ has the same form as $\ket{\psi}$ on $\mc{C}$.
\end{proofStep}

\begin{proofStep}
\label{step:thm:conscious-observer:preserve}
Now I will prove the existence of a state $\ket{\psi'}$ in which the observer has the same structure but the property $\oper{A}_{\varepsilon}$ of the environment has the value $a'$, and the two states $\ket{\psi}$  and $\ket{\psi'}$ can be related by a unitary transformation $\oper{S}$ that commutes with the Hamiltonian.

In \cite{Stoica2023PrinceAndPauperQuantumParadoxHilbertSpaceFundamentalism} Lemma 1 it was shown that such a symmetry transformation exists, if the observable $\oper{A}_{\varepsilon}$ has eigenspaces of the same dimension, and for every eigenvalue $a$, $-a$ is also an eigenvalue. Then, there is a unitary transformation $\oper{S}$ commuting with $\oper{H}$ so that $\oper{S}\ket{\psi}=\ket{\psi'}$.
This extends easily to a generic observable $\oper{A}_{\varepsilon}$ and two eigenvalues $a\neq a'$ of different multiplicities.
We can always choose instead of $\oper{A}_{\varepsilon}$ another observable $\oper{\wt{A}}_{\varepsilon}$ having the same eigenspaces, but different eigenvalues, so that $a=-a'$.
Then if a state is an eigenstate of $\oper{\wt{A}}_{\varepsilon}$, it is also an eigenstate of $\oper{A}_{\varepsilon}$, and we can deduce the eigenvalue of $\oper{A}_{\varepsilon}$ from the eigenvalue of $\oper{\wt{A}}_{\varepsilon}$.
Depending on whether the eigenspace of $\oper{A}_{\varepsilon}$ corresponding to the eigenvalue $a$ is smaller, equal, or greater than the eigenspace corresponding to the eigenvalue $a'$, there are three cases.
By a space ``smaller'' than another space I understand here that there is a unitary transformation mapping the first space into a strict subspace of the second space.
If the eigenspaces are equal, we can apply Lemma 1 from \cite{Stoica2023PrinceAndPauperQuantumParadoxHilbertSpaceFundamentalism}.
If one of them is smaller than the other, we choose the observable $\oper{\wt{A}}_{\varepsilon}$ to commute with $\oper{A}_{\varepsilon}$ and so that the states $\ket{\psi}$ and $\ket{\psi'}$ are contained in eigenspaces of $\oper{\wt{A}}_{\varepsilon}$ of the same dimension.
To make sure that $\oper{\wt{A}}_{\varepsilon}$ contains the necessary information to find out the eigenvalues of $\oper{A}_{\varepsilon}$ from those of $\oper{\wt{A}}_{\varepsilon}$, we can choose $\oper{\wt{A}}_{\varepsilon}$ to be finer, in the sense that all its eigenspaces are included in eigenspaces of $\wt{A}$.
Therefore, we can apply Lemma 1 from \cite{Stoica2023PrinceAndPauperQuantumParadoxHilbertSpaceFundamentalism}, and obtain that there is a structural symmetry transformation $\oper{S}$ mapping $\ket{\psi}$ to $\ket{\psi'}$.
We don't need to perform the measurement, the whole point is only to prove the existence of a structural symmetry transformation $\oper{S}$ mapping $\ket{\psi}$ to $\ket{\psi'}$.
\end{proofStep}

\begin{proofStep}
\label{step:thm:conscious-observer:knows-nothing}
We need to find out the probability measure for the possible alternative wavefunctions $\psi_{\varepsilon}(\q'_{\varepsilon})$ of the environment resulting from a structural symmetry transformation from $\psi$.
Since with any wavefunction $\psi_{\varepsilon}$ all other wavefunctions $\psi'_{\varepsilon}$ are equally present on different parameter spaces, this measure has to be the uniform measure invariant to the allowed unitary transformations.
Any other measure would break the unitary symmetry of the state space, and would introduce preferred parameter spaces by a sleight of hand.

Let us find out the probability as the ratio between the measure of ``favorable cases'' and the measure of ``all possible cases''.
Denote by $\mu_{\approx}$, respectively $\mu_{\not\approx}$, the measure of the set of unit vectors $\ket{\psi'_{\varepsilon}}\in\hilbert_{\varepsilon}$ resulting from a structural symmetry transformation from $\psi$ that are consistent, respectively inconsistent with the observer's memory.
The probability that the memory of the observer-like structure contains accurate data about the external world is
\begin{equation}
\label{eq:probability-knowledge}
\mu_{\approx}/(\mu_{\approx}+\mu_{\not\approx}).
\end{equation}

As seen above, the measure $\mu$ is uniform. This implies that the probability from equation \eqref{eq:probability-knowledge} is the same as the probability that the observer-like structure guesses the values of the properties of the external world by pure chance. In other words, \emph{there is no correlation between the observer's state and the environment, so the observer knows nothing about the external world}.

Let us take a more concrete look.
The observer-like structure's memory encodes the knowledge that the property $A_{\varepsilon}$ has a definite value $a$.
The same physical property $A_{\varepsilon}$ relative to $\mc{C}'$ is represented by the observable $\oper{A}'_{\varepsilon}=\oper{S}\oper{A}_{\varepsilon}\oper{S}^\dagger$. Therefore, on $\mc{C}'$, the property of $A_{\varepsilon}$ is $a'$, not $a$, as the observer-like structure would think. The observable $\oper{A}_{\varepsilon}$ represents a different physical property on $\mc{C}'$, whose value is indeed still $a$, but as known from Proposition \ref{thm:alternative-positions} and Principle \ref{pp:meta-relativity-structure}, for observer-like structures on different parameter spaces, the same observable represents a different physical property. So the observer-like structure on $\mc{C}'$ associates $A_{\varepsilon}$ with the value $a$, but its parameter space $\mc{C}'$ associates it with the value $a'$.

Because the same observer-like structure $\psi_{\omega}(\q_{\omega})$ can have any possible environment $\psi_{\varepsilon}(\q'_{\varepsilon})=\psi'_{\varepsilon}(\q_{\varepsilon})$, with uniform probability, its memory contains zero information about the value of $\oper{A}_{\varepsilon}$.

However, perhaps we shouldn't consider as ``all possible cases'' all possible environments resulting from a structural symmetry transformation from $\psi$, but only those for which $\oper{A}_{\varepsilon}$ has a definite value.
Even restricted like this, the probability still has to be uniform, because together with any eigenvector $\ket{\psi_{\varepsilon}}\in\hilbert_{\varepsilon}$, any other eigenvector $\ket{\psi'_{\varepsilon}}\in\hilbert_{\varepsilon}$ is equally obtainable by a structural symmetry transformation.
We see that if all structures identical with the observer's structure were conscious regardless of their parameter space, the probability \eqref{eq:probability-knowledge} would be the same, so the observer would know about the property represented by $\oper{A}_{\varepsilon}$ exactly what is allowed by random guess.
\end{proofStep}

So indeed if observers were reducible to structures, they would know nothing about the external world.
\end{proof}

\begin{remark}
\label{rem:change-property}
One may think that on $\mc{C}'$ we should consider the observable $\oper{A}_{\varepsilon}$ instead of $\oper{A}'_{\varepsilon}=\oper{S}\oper{A}_{\varepsilon}\oper{S}^\dagger$. Then, $\oper{A}_{\varepsilon}\ket{\psi_{\varepsilon}}=a\ket{\psi_{\varepsilon}}$ remains true regardless of the basis, and it may seem that the knowledge encoded by the observer-like structure on $\mc{C}'$ is valid. But we should not forget that the starting point was that the observer-like structure determines the physical meaning of the observables. For example, if $\oper{A}_{\varepsilon}$ represents on $\mc{C}$ the coordinate $x$ of an external object, and its value is $a$, on $\mc{C}'$ the coordinate $x$ of the object appears to be $a'$, and it is represented by $\oper{A}'_{\varepsilon}$, not by $\oper{A}_{\varepsilon}$. Indeed, $\oper{A}'_{\varepsilon}\ket{\psi_{\varepsilon}}=a'\ket{\psi_{\varepsilon}}$, as Proposition \ref{thm:alternative-positions} says.
The observable $\oper{A}_{\varepsilon}$ will have the same value, but it will represent another property, and not the position of the object.
\qed
\end{remark}

\begin{remark}
\label{rem:thm:conscious-observer-classical}
In this article, by ``observer'' I don't necessarily mean somebody making a quantum measurement. Although any observer inevitably makes quantum observations, the result doesn't rely on the measurement problem.
The following Corollary should clear any doubt.
\qed
\end{remark}

\begin{corollary}
\label{thm:conscious-observer-classical}
Theorem~\ref{thm:conscious-observer} is true in a classical world too.
\end{corollary}
\begin{proof}
Koompan and von Neumann showed how to formulate Classical Physics using the quantum formalism \cite{Koopman1931HamiltonianSystemsAndTransformationInHilbertSpace,vonNeumann1932KoopmanMethod}.
But there are some differences.
The momentum operators $\wh{p}_j$ are not the form $-i\hbar\frac{\partial}{\partial q_j}$, they are independent of $\wh{q}_j$ and commute with them.
The resulting parameter space $\mc{C}$ is the classical \emph{phase space}, containing both $q_j$ and $p_k$ as generalized coordinates.
The physical states are represented only by basis vectors $\ket{\q,\p}$. The wavefunctions are localized at points $(\q,\p)\in\mc{C}$.
On the state vectors $\ket{\q,\p}$, all observables $\wh{f}(\wh{\q},\wh{\p})$ have definite values.
Because the state always remains classical, the evolution operators $\oper{U}_t$ from equation \eqref{eq:unitary_evolution} always map basis vectors to basis vectors.
The symmetry transformations of the phase space, called \emph{canonical transformations}, are represented only by unitary transformations that map basis vectors to basis vectors.

Since these restrictions don't affect the proof of Theorem~\ref{thm:conscious-observer}, the result applies to classical worlds too.
\end{proof}

\begin{corollary}
\label{thm:conscious-observer-maximal}
The parameter space supporting observers is unique up to spacetime and gauge symmetries.
\end{corollary}
\begin{proof}
The proof of Theorem~\ref{thm:conscious-observer} shows that only in some parameter spaces the wavefunction $\psi_{\varepsilon}$ of the external world corresponds to the memory of the observer represented by $\psi_{\omega}$.
Any unitary transformation $\oper{S}$ as in the proof of Theorem~\ref{thm:conscious-observer} introduces an ambiguity in some physical properties.
Then, as in the proof of the Theorem, the observer is unable to know these properties.
The observer's memory contains information about a very limited number of properties, so it doesn't fix the wavefunction $\psi_{\varepsilon}$ of the external world.
Consequently, the parameter space $\mc{C}_{\varepsilon}$ is far from being completely determined.
However, the observer can make experiments to determine any physical property of a system. This means that all observable properties should be accessible to the observer, given the right experiments.
Therefore, the properties that can't be known to the observer even in principle have to be ``nonphysical'', \ie dependent on the reference frame or the gauge.
The properties that can be known identify a unique parameter space supporting observers, up to spacetime and gauge symmetries.
\end{proof}

Theorem~\ref{thm:conscious-observer} answers Question~\ref{question:observers-structure} negatively, and, based on this, Corollary~\ref{thm:conscious-observer-maximal} answers Question~\ref{question:correspondence-observables-meaning} affirmatively.

\section{Discussion}
\label{s:discussion}

We take for granted the existence of a unique correspondence between the operators representing properties, and the physical properties themselves.
But we have seen that, if all observer-like structures were observers, experiments couldn't ensure uniquely this correspondence.
This correspondence is absent from the theoretical description too.
It only seems to be part of the theory because we give different names to the various operators, we label them with different symbols, and we all follow the convention, so we agree with one another about their meaning.
The uniqueness of this correspondence is ensured by the fact that only observer-like structures from a particular parameter space can be observers, which is proven in Theorem~\ref{thm:conscious-observer}.
Without this, no observer would be able to know anything about the external world.

We also take for granted that our memory holds correct information about the external world automatically, just because we interacted with it in the past.
This predisposition may make the proof of Theorem~\ref{thm:conscious-observer} more difficult to understand.
But the evolution equations of physics are reversible, and if we remember our past interactions, we should equally remember our future interactions. Or rather there should be no relation between the content of our between memory and the external world at all, because all state vectors $\ket{\psi_{\omega}}\ket{\psi_{\varepsilon}}$ are equally ``legal'' under the laws of physics.
The states containing brains with memories that don't correspond to facts about the external world are as ``legal'' as those with reliable memories, and even overwhelmingly outnumber them.
Without special conditions that ensure the reliability of our memories, most observers would fluctuate ephemerally into existence by accident and then dissipate (\cite{Eddington1934NewPathwaysInScience_Boltzmann_brains}, p. 65).
The chance of not being a \emph{Boltzmann brain} would be practically zero. 

Fortunately, the initial state of the universe was extremely special. Penrose estimated how special it was: one in $10^{10^{123}}$ \cite{Penrose1989EmperorsNewMind}.
The Boltzmann entropy was extremely low, and it increases steadily in time, ensuring the validity of the Second Law of Thermodynamics \cite{Boltzmann1964LecturesOnGasTheory,Feynman1965TheCharacterOfPhysicalLaw}.
This is believed to also ensure the relative reliability of our memories.
And even though, as shown in \cite{Stoica2022DoesQuantumMechanicsRequireConspiracy}, reliable memories require more fine tuning of the initial state than just starting in a low-entropy state, we live in such a friendly universe.
Our memories about the properties of external systems are reliable because we are part of a universe in such a select state.

But this is not true of most states in which the universe could be, including most states resulting from its actual state by unitary transformations $\oper{S}$.
No matter how friendly our universe appears on our parameter space, it will not appear friendly at all to the observer-like structures from most alternative parameter spaces.
Therefore, if observers were reducible to structures, any observer should expect that in the very next moment the universe containing it will turn out to be crazy. There would be rare instances when the observer-like structure survives for a brief period of time, and even then, in most cases, it would experience a surrealist reality.
Every time when this doesn't happen to us is a subtle reminder that we are more than the structure.

\section{Implications}
\label{s:implications}

It is sometimes believed that the state vector and the Hamiltonian are sufficient to recover everything else, the space structure, the tensor product structure (necessary for the existence of subsystems), and the correspondence between observables and physical properties. This seems to be needed by various quantum-first approaches to Quantum Gravity, see \cite{Carroll2021RealityAsAVectorInHilbertSpace} and other references in \cite{Stoica2021SpaceThePreferredBasisCannotUniquelyEmergeFromTheQuantumStructure}. 
In \cite{Stoica2021SpaceThePreferredBasisCannotUniquelyEmergeFromTheQuantumStructure,Stoica2022VersatilityOfTranslations,Stoica2024DoesTheHamiltonianDetermineTheTPSAndThe3dSpace} it was shown that this is impossible, and there are infinitely many physically distinct but isomorphic structures.
Any quantum-first approach is extremely ambiguous, resulting in infinitely many solutions representing physically different worlds.
But we could hope that this ambiguity is harmless and consistent with the empirical data.
However, Theorem~\ref{thm:conscious-observer} shows that this isn't true. If it were true, the observer-like structures would not be able to know anything about the external world. This would contradict the most basic empirical facts and would make science impossible.

Theorem~\ref{thm:conscious-observer} also rejects the thesis that physical properties are purely relational, as proposed by Rovelli \cite{Rovelli1996RelationalQM}, since this also implies that observers would know nothing about the external world.

Proposals like the above may come from the implicit assumption that \emph{ontic structural realism}, the thesis that only the structure exists, that there are only relations and no relata, is true \cite{sep-structural-realism}. This is also endorsed by materialism or physicalism, positions that don't admit ontology having phenomenal (that is, sentiential or experiential) powers. Theorem~\ref{thm:conscious-observer} shows otherwise: not all isomorphic structures are created equal.
Epistemic structural realism proposes that even if things have a nature of their own, this is inaccessible to us through science (see Section~\sref{s:meta-relativity-observers}).
But Theorem~\ref{thm:conscious-observer} shows that, without grounding our knowledge into something in addition to the structures, we would know nothing.
Something makes only some of the isomorphic observer-like structures be observers.
So whatever breaks structural realism, this is manifest through the observers, to the observers.

But what is an observer?
In this article I had in mind the human observers as a directly verifiable example familiar to all of us. An observer-like structure is any structure isomorphic with the structure of an observer. But to restrict observers to humans would be anthropomorphism, so if there is a way to characterize observers in a non-anthropomorphic way, we should adopt it. But whatever the definition of an observer is, such an observer must have a structure. And regardless of the characteristic of the structure of the observers, Theorem~\ref{thm:conscious-observer} shows that not all structures isomorphic with it are observers.

This article makes no claim to define or elucidate what kind of structure a system must have to be a conscious observer. This problem belongs to other fields, from Neuroscience to philosophy of mind.
But the results from this article inform these fields that observers, whatever they are, are not reducible to their structure.
For example, the \emph{computational theory of mind} proposes that the mind is reducible to a computation \cite{sep-computational-mind,ColomboPiccinini2023TheComputationalTheoryOfMind}.
If ``computation'' means what is understood in Computer Science, this is already rejected \cite{Stoica2023DoesAComputerThinkIfNoOneIsAroundToSeeIt}, and if it also means the internal structure of the machine implementing it, this would contradict Turing universality \cite{Stoica2023DoesAComputerThinkIfNoOneIsAroundToSeeIt} and Theorem~\ref{thm:conscious-observer}.
\emph{Functionalism} \cite{sep-functionalism} proposes that the mind reduces to the way it functions. If by ``function'' we understand structure and its dynamics, this is in conflict with Theorem~\ref{thm:conscious-observer}.
\emph{Illusionism} proposes that \emph{phenomenal consciousness}, experience itself, is an illusion of the computation or function of the structure \cite{Dennett2016Illusionism,Frankish2016IllusionismAsATheoryOfConsciousness}. 
Even if we go down to the brain's finest structural details, as in the \emph{identity theory} \cite{sep-mind-identity}, we can't avoid Theorem ~\ref{thm:conscious-observer}.
The structure of a Carbon atom or even of any particle can exist in any possible parameter space, hence the problem remains.
In this sense, Theorem~\ref{thm:conscious-observer} shows that the observer-like structures from other parameter spaces are \emph{philosophical zombies} \cite{KirkSquires1974ZombiesVMaterialists,sep-zombies}, disproving thus the materialist thesis that the ontological substrate can't have phenomenal effects.
On the other end of the spectrum we find the proposals that the mind is not reducible to structure.
\emph{Panpsychism} proposes that even the elementary particles have such mental properties \cite{sep-panpsychism}. A naive rejection of panpsychism is that it adds new properties unknown in physics, and this should lead to different predictions than, for example, Particle Physics. But this article shows that such properties correspond in fact to the already known physical properties.
Corollary~\ref{thm:conscious-observer-maximal} shows that this correspondence has to be unique (up to Poincar\'e and gauge symmetries), and it should go down to the complete set of basis observables. This implies a full identification between mental and physical properties, suggesting a form of monism. \emph{Neutral monism} proposes that the intrinsic nature of things appears externally as physical properties, and internally as mental properties \cite{Russell1927AnalysisOfMatter}. \emph{Idealism} \cite{sep-idealism,Kastrup2019AnalyticIdealismAConsciousnessOnlyOntology,Stoica2020NegativeWayToSentience,Indich1995ConsciousnessInAdvaitaVedanta} is a monistic position that identifies the physical properties and the physical laws as the structure and dynamics of a fundamental consciousness.
Another position is \emph{dualism} \cite{sep-dualism}, stating that both matter and mind are fundamental and either interact or mirror one another. This would unnecessarily duplicate both the ontology and the structures, so it would be redundant. A monistic position wouldn't have this problem.
Whatever the explanation is, it should take into account that observers are not reducible to their structure. Corollary~\ref{thm:conscious-observer-maximal} shows all physical properties are grounded in sentience, in a way that goes beyond the experience of individual observers, making them agree with one another and with the external world about the physical meaning of the observables.
There is an essentially unique parameter space supporting observers, and observer-like structures from other parameter spaces are philosophical zombies.
This difference seems to elude both the theoretical description and the intersubjectively verifiable experiments.
However, as shown by Theorem~\ref{thm:conscious-observer} and Section~\sref{s:discussion}, this is revealed empirically by the fact that at any instant we could turn out to be observers in a crazy surrealistic world completely unrelated to our memories, but every time we find ourselves in a friendly one.
As if the universe is so friendly that it reassures us at every instant about this fact.
Indubitably, structure remains important, and trying to characterize the structure of conscious systems is essential in advancing our understanding of observers.

The existence of an ontologically special basis beyond structure and relations was conjectured previously because it allows reasoning about the self-location of the observer in a way that leads to the Born rule \cite{Stoica2022BornRuleQuantumProbabilityAsClassicalProbability} and endows the Many-Worlds Interpretation with genuine probabilities and a local ontology \cite{Stoica2023TheRelationWavefunction3DSpaceMWILocalBeablesProbabilities}. Other proposals that don't use a fixed ontic basis fail to get the Born rule (\cite{Stoica2023TheRelationWavefunction3DSpaceMWILocalBeablesProbabilities}, Proposition 1, \S 6).
Theorem~\ref{thm:conscious-observer} justifies conjecturing such ontic differences, by showing that this is unavoidable.


\end{document}